\documentclass[onecolumn,10pt]{IEEEtran}
\usepackage{amsfonts,amsmath,amssymb,amsthm,caption}
\usepackage[ruled,vlined,commentsnumbered]{algorithm2e}
\usepackage{algpseudocode}
\usepackage{ulem}
\usepackage{amsbsy}
\usepackage{graphicx,multirow,bm}
\usepackage{color}
\usepackage{subfigure}
\makeatletter
\newtheoremstyle{mythm}{3pt}{3pt}{}{16pt}{\bfseries}{:}{.5em}{}
\theoremstyle{mythm}
\newtheorem{theorem}{Theorem}
\newtheorem{example}{Example}
\newtheorem{definition}{Definition}
\newtheorem{remark}{Remark}

\newtheorem{lemma}{Lemma}
\newtheorem{construction}{Construction}

\newcommand{\tabincell}[2]{\begin{tabular}{@{}#1@{}}#2\end{tabular}}
\begin{document}
\title{On Secure Coded Caching via Combinatorial Method}
\author{Minquan Cheng, Dequan Liang, Ruizhong Wei

\thanks{Cheng and Liang are with Guangxi Key Lab of Multi-source Information Mining $\&$ Security, Guangxi Normal University,
Guilin 541004, China, (e-mail: chengqinshi@hotmail.com, dequan.liang@hotmail.com).}
\thanks{R. Wei is with Department of Computer Science, Lakehead University, Thunder Bay, ON, Canada, P7B 5E1,(e-mail: rwei@lakeheadu.ca).}}

\date{}
\maketitle

\begin{abstract}
Coded caching is an efficient way to reduce network traffic congestion during peak hours by storing some content at the user's local cache memory without knowledge of later demands. The goal of coded caching design is to minimize the transmission rate and the subpacketization. In practice the demand for each user is sensitive since one can get the other users' preferences when it gets the other users' demands. The first coded caching scheme with private demands was proposed by Wan et al. However the transmission rate and the subpacketization of their scheme increase with the file number stored in the library. In this paper we consider the following secure coded caching: prevent the wiretappers from obtaining any information about the files in the server and protect the demands from all the users in the delivery phase. We firstly introduce a combinatorial structure called secure placement delivery array (SPDA in short) to realize a coded caching scheme for our security setting. Then we obtain three classes of secure schemes by constructing SPDAs, where one of them is optimal. It is worth noting that the transmission rates and the subpacketizations of our schemes are independent to the file number. Furthermore, comparing with the previously known schemes with the same security setting, our schemes have significantly advantages on the subpacketizations and for some parameters have the advantage on the transmission rates.
\end{abstract}

\begin{IEEEkeywords}
Coded caching, private demands, Secure placement delivery array, transmission rate, subpacketization
\end{IEEEkeywords}

\section{Introduction}
The wireless network is increasingly suffered from the heavy network traffic since a steep increase of wireless devices connect to the internet due to the multimedia streaming, web-browsing and social networking. Furthermore, the high temporal variability of network traffic results in congestion during peak traffic times and underutilization of the network during off-peak traffic times. Caching can effectively shift traffic from peak to off-peak times \cite{BBD}. That is, fractions of popular content are stored locally in each user's cache memory across a given network. During the peak traffic times, the users can be partly served from their local cache, thereby reducing the network load. The goal is to minimize the amount of transmission in delivery phase.

\subsection{Coded caching system}
In order to further reduce the transmission amount based on the local caches during the peak traffic times, coded caching strategy was originally proposed in \cite{MN} for the shared-link broadcast networks where a single server containing $N$ files is connected to $K$ users through a shared link. Each user has a cache memory which can store $M$ files (see Fig. \ref{system}).
\begin{figure}[h]
\centering\includegraphics[width=0.4\textwidth]{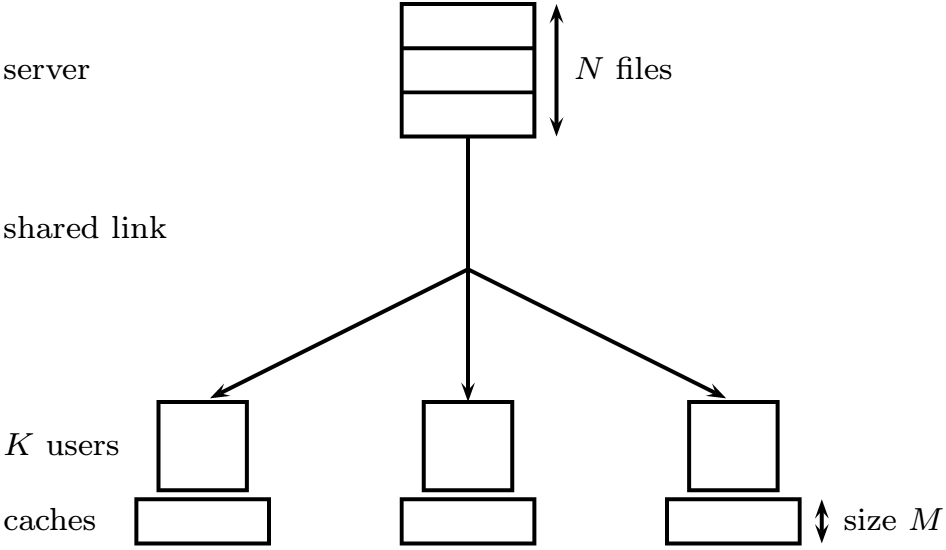}
\caption{Caching system}\label{system}
\end{figure}

An $F$-division $(K,M,N)$ coded caching scheme originally introduced in \cite{MN} consists of two phases.
\begin{itemize}
\item Placement phase: During the off-peak traffic times, the server divides each file into $F$ packets, each of which has the same size, and then places some uncoded/coded packets of each file into users' cache memories. $F$ is usually referred as subpacketization. In this phase the server does not know the users' requests in the following phase.
\item Delivery phase: During the peak traffic times, each user randomly requests one file from $N$ files independently. Then the server broadcasts some coded signals (XOR of some required packets) such that each user's demand is satisfied with the help of the contents of its cache.
\end{itemize}

Let $R$ denote the maximum transmission amount among all the request during the delivery phase. $R$ is called the rate of a coded caching scheme. Furthermore the complexity of implementing a coded caching scheme increases with the subpacketization $F$. Naturally we prefer to design a coded caching scheme with the transmission rate and the subpacketization as small as possible.

Maddah-Ali and Niesen in \cite{MN} proposed the first deterministic scheme, i.e., for any positive integers $K$, $M$ and $N$ with $M<N$, if $KM/N$ is an integer, there exists a ${K\choose KM/N}$-division $(K,M,N)$ coded caching scheme. Such a scheme is referred to as the MN scheme in this paper. In \cite{WTP} the authors showed that the MN scheme has minimum rate by the graph theory under uncode placement. Hence the MN scheme has also been widely used to various settings such as D2D networks \cite{JCM}, hierarchical networks \cite{KNMD} etc.

\subsection{Secure coded caching}
Secure coded caching was originally considered in \cite{STC}, where there are some wiretappers who can receive the broadcasted packets from the server in the delivery phase. To prevent the wiretappers from obtaining any information about the files in the server, each user stores not only the content about the files in its cache, but also some secret keys. In the delivery phase, each coded signal is encrypted by some key such that only the authorized users can decrypt it. This secure caching scheme against wiretappers was proved in \cite{BATV} to be optimal under the constraint of uncoded cache placement. Another secure shared-link caching model was proposed in \cite{RPKP} where the objective is to avoid each user to get any information about the files not required by that user. The placement and delivery phases were designed based on the MN caching scheme with an additional secret sharing precoding \cite{Shamir} on each file. It is worth noting that the secure caching scheme in \cite{RPKP} could also successfully prevent external wiretappers due to the secret sharing. According to the idea of Private Information Retrieval (PIR) problem in \cite{CGKS}, the private demand was first proposed by Wan and Caire in \cite{WC}. When all the users just require a signal distinct files, the authors proposed an ordered scheme with private demands. Concretely speaking, the authors in \cite{WC} slightly modified the MN placement strategy in the placement phase, i.e., by supplementing $N-K$ virtual users and using the MN delivery strategy to protected each user's demand, in the delivery phase. We list this scheme with the user number $K$, memory ratio $M/N$, file number $N$, subpacketization $F$ and transmission rate $R$ in Table \ref{tab-known-1}. When $N=K=2$, Kamath in \cite{Kamath} improved the lower bound of the transmission rate in \cite{STC}. The above strategies to prevent external wiretappers and internal malicious users from retrieving information about the files, were also used in extended models, such as D2D caching systems \cite{ZY,Mathar}, topological cache-aided relay networks \cite{ZY2}, erasure broadcast channels \cite{KSWO}, etc.
{\begin{table*}[!htbp]
\center
\caption{The known secure coded caching schemes\label{tab-known-1}}
\small{
\begin{tabular}{|c|c|c|c|c|c|c|}
\hline
Parameters & File Number&User number & Cache Fraction $\frac{M}{N}$
& Rate $R$   & Subpacketization $F$    \\ \hline
\tabincell{c}{\cite{WC}: $N,t>K\in Z^{+}$ \\ with $K+1<t<N$}  &$N$& $K$ &$\frac{{N-1\choose t-1}-{N-K-1\choose t-K-1}}{{N\choose t}-{N-K-1\choose t-K-1}}$&$\frac{{N\choose t+1}}{{N\choose t}-{N-K-1\choose t-K-1}}$&${N\choose t}-{N-K-1\choose t-K-1}$\\ \hline
\end{tabular} }
\end{table*}}

Considering that each user can get any file if it sends the demand to the server, in this paper we do not consider the security of files for each user. In practice the demand for each user is sensitive since one can get the other users' preferences when it gets the other users' demands. So in this paper we will consider the following security setting: prevent the wiretappers from obtaining any information about the files in the server and protect the demands from all the users in delivery phase.

It is worth noting that since each user caches some secret keys in the delivery phase, and the server sends each encrypted coded signal by some key in the delivery phase of the scheme in \cite{WC}, the new scheme can against wiretappers. However from Table \ref{tab-known-1}, we can see that the subpacketization ${N\choose t}-{N-K-1\choose t-K-1}$ approximates the MN subpacketization ${N\choose t}$. This implies that the scheme will change along with the value of $N$. Intuitively for the fixed user number $K$ and memory ratio $M/N$, performance of the scheme in Table \ref{tab-known-1} get worse with the increasing value of file number $N$. This fact will be shown by an example in Subsection \ref{sec-comparison}. So given the user number and memory ratio, it is meaningful to design a coded caching scheme satisfying private demands such that the transmission rate and the subpacketization are independent to the value of $N$.
\subsection{Contribution and organizations}
In this paper, according to the formulation in \cite{STC}, we first reformulate the shared-link coded caching model, and the constraints on the privacy of the users' demands from other users. Then based on our formulation, we characterize the placement and delivery phases for a coded caching scheme with private demands by a interesting combinatorial structure which is called secure placement delivery array (SPDA in short). This structure has a close relationship with placement delivery array proposed by Yan et al. in \cite{YCTC}. Secondly the lower bound of the minimum transmission of the scheme which can be realized by a SPDA is derived. According to our lower bound, we construct the scheme with $K$ users achieving such lower bound for any positive integer $t$, $K$ with $(t+1)|K$. Based on such construction, the constructions of our schemes with good performance for the other parameters $t$ and $K$ are proposed. We list these results in Table \ref{tab-main}.
{\begin{table*}[!htbp]
\center
\caption{The main results in this paper\label{tab-main}}
\small{
\begin{tabular}{|c|c|c|c|c|c|c|}
\hline
parameters: $t,K\in Z^{+}$ &User number & Cache Fraction $\frac{M}{N}$
&Rate $R$   & Subpacketization $F$    \\ \hline

Theorem \ref{th-SPDA-optimal}:  $(t+1)|K$  & $K$ &$\frac{K}{2K-t}$&$\frac{K-t}{t+1}\cdot\frac{K}{2K-t}$&${K\choose t}+{K-1\choose t}$\\ \hline

\tabincell{c}{Theorem \ref{th-SPDA-2}: $(t+1)\nmid K$, $\lfloor K/2\rfloor$\\  $>(t+1)$. $K'=\min\{K''|K''$ \\ $>K, (t+1)|K'',K''\in Z^{+}\}$} & $K$ &$\frac{{K-1\choose t-1}+{K'-1\choose t}}{{K\choose t}+{K'-1\choose t}}$&$\frac{{K\choose t+1}}{{K\choose t}+{K'-1\choose t}}$&${K\choose t}+{K'-1\choose t}$\\ \hline

\tabincell{c}{Theorem \ref{th-SPDA-3}: $t,K\in Z^{+}$ \\ with $(t+1)\nmid K$,$\lfloor K/2\rfloor<(t+1)$} & $K$ &$1-\frac{(K-t)(t+1)}{K(K+1)}$&$\frac{K-t}{K+1}$&${K\choose t}+{K\choose t+1}$\\ \hline
\end{tabular} }
\end{table*}}

Clearly in Table \ref{tab-main}, the transmission rates and the subpacketizations of our schemes are independent to the file number $N$. So when $N$ is far larger than $K$, our scheme have significant advantages on the subpacketization. In Subsection \ref{sec-comparison}, by an example we show that our transmission rate is close to the transmission rate of the scheme in Table \ref{tab-known-1}. Furthermore, for some memory ratio, our transmission rate is smaller than that of the scheme in Table \ref{tab-known-1}

 The rest of this paper is organized as follows. Section \ref{System Model} introduces coded caching scheme with private demands and its related concepts. In Section \ref{Sec-Coded caching scheme-SPDA}, SPDA is proposed and then its connection with coded caching scheme with private demands is established. In Section \ref{Lower bounds} the lower bound of the transmission rate of the schemes realized by SPDAs is derived. According to our lower bound, three classes of schemes with private demands are constructed in  Section \ref{sec-construciton}, where one class of the schemes achieves our lower bound. Furthermore some comparisons with the scheme in \cite{WC} are proposed. Finally conclusion is drawn in Section \ref{conclusion}.

\section{System Model with security}
\label{System Model}
In this paper, we will use the following notations unless otherwise stated. For any positive integers $K$ and $t$ with $t<K$, let $[0,K)=\{0,1,\ldots,K-1\}$ and ${[0,K)\choose t}=\{A\ |\   A\subseteq [0,K), |A|=t\}$, i.e., ${[0,K)\choose t}$ is the collection of all $t$-sized subsets of $[0,K)$.

In what follows, when we say that an event is randomly selected, it means that it is uniform randomly selected. Therefore,
for convenience, we can use the names of domain sets instead of the names of random variables over the domains in entropy and mutual information formula. In our system model, as usual, a single server containing $N$ files connects to $K$ users over a shared link, and each user has a cache memory with size of $M$ files. We also assume that there is a secure channel between the server and each of the users. But
the usage of these secure channels are very limited in the delivery phase. Denote the $N$ files by $\mathcal{W}$ $=\{{W}_0$, ${W}_1$, $\ldots$, ${W}_{N-1}\}$, $K$ users by $\mathcal{K}=[0,K)$ and the requested file numbers by ${\bf d}=(d_0,d_1,\cdots,d_{K-1})$, i.e., user $k$ requests the file ${W}_{d_k}$ where $d_k\in[0,N)$ for each $k\in \mathcal{K}$. The two phases of an $F$-division $(K,M,N)$ secure coded caching scheme can be reformulated as follows.
\begin{itemize}
\item {\bf Placement Phase:} During the off peak traffic times, each file is divided into $F$ equal packets, i.e., $W_{n}=\{W_{n,j}\ |\ j=0,1,\cdots,F-1\}$, $0\leq n<N$. Assume that each file has $B$ bits, i.e., each file $W_{i}$ is randomly chosen from $[0,2)^B$. The server randomly choose $N$ different permutation functions on the $[0,F)$, say $\Omega=\{\phi_{n}\ |\ n\in [0,N)\}$. These permutation functions are only known to the server at this phase. For each user $k$, the server stores $\mathcal{Z}_{k} \bigcup \mathcal{T}_k$ to the user's cache, where $\mathcal{Z}_{k}$ contains a certain number of the packets of file $W_n, n \in [0,N)$
whose indexes have been changed by the permutation functions $\phi_{n}$. The size of $\mathcal{Z}_k$ is less than or equal to the capacity of each user's cache memory size $MB$, where $0<M<N$. $\mathcal{T}_k$ consists of $F(1-\frac{M}{N})$ secret keys. In order to prevent the wiretapper from obtaining any information about the contents transmitted in delivery phase, similar to Sengupta et al. in \cite{STC}, all the secret keys are chose mutually independent and randomly from $[0,2)^{B/F}$. The server creates certain subgroups of users so that each signal would be sent to one of the groups.
 Assume that each group of users gets at most one multicast opportunity. Then each user $k$ will belong to $F(1-\frac{M}{N})$
 groups and have $F(1-\frac{M}{N})$ keys. We use $\mathfrak{U}_k$ to represent the indices of these secret keys cached by user $k$. That is, for each key, say $ T_{\mathcal{U}}$, all the users in group $\mathcal{U}$ have such secret key and they will receive a coded signal in the delivery phase which is encrypted by such secret key. Let $\mathfrak{U}=\bigcup_{k\in[0,K)}\mathfrak{U}_k$. In this phase, any user can get none information of any permutation. Actually the number of files is large and it is much larger than the number of the users. It is worth noticing that for each user, the size of contents for popular files cached by the user is much larger than the size of secret keys cached by the user. Hence we don't consider the cost of the cache occupied by the secret keys in the system.

From the above introduction, we have that for any $\phi _n, n\in [0.N)$ and any user $k \in [0,K)$, the following equation holds.
 {$$ I(\{ \phi_n\}; \mathcal{T}_k, \mathcal{Z}_k,\mathfrak{U}_k) = 0.$$}
\item {\bf Delivery Phase:} During the delivery phase, we assume that each user independently demands one file randomly. The users' requests are go over secret channels. The requirements of all the users are denoted by a vector $\mathbf{d}=(d_{0},d_{1},\ldots,d_{K-1})$. According to the cached content of all $K$ users $\mathcal{Z}=\{\mathcal{Z}_{0},\mathcal{Z}_{1},\ldots,\mathcal{Z}_{K-1}\}$ and the demand vector ${\bf d}$, the server generates a coded signal $T_{\mathcal{U}} \bigoplus X_{\mathcal{U}}$, $\mathcal{U}\subseteq [0,K)$ for each $\mathcal{U}\in \mathfrak{U}$. Here $X_{\mathcal{U}}$ is obtained by XOR some packet which is useful for the users in $\mathcal{U}$. It is worth noting that the indices of these packets have been carried by the corresponding permutations. In order to make each user in $\mathcal{U}$ receive successfully the coded signal $T_{\mathcal{U}} \bigoplus X_{\mathcal{U}}$ and use the secret key $T_{\mathcal{U}}$,
 the server should transmit $T_{\mathcal{U}} \bigoplus X_{\mathcal{U}}$ with its header which represents the group $\mathcal{U}$ and the indices of the packets from $X_{\mathcal{U}}$.
 Clearly the header of $T_{\mathcal{U}} \bigoplus X_{\mathcal{U}}$ is independent to the bits of $T_{\mathcal{U}}$ and $X_{\mathcal{U}}$. We assume that the total length of the headers is negligible compared to the file length $B$ such that, for simplicity, the relevant amount of transmission is the size of all the transmitted coded signals.
 By receiving the signals from the server and with the help of the cached content, each user $k$ can decode all the $F$ packets $W_{d_{k},\phi_{d_{k}}(j)}$, $j\in[0,F)$, $\phi_{d_k}\in \Omega$ of the file $W_{d_k}$. Clearly each user $k$ just obtains all the packets of the file $W_{d_{k}}$ but does not know how to reconstruct $W_{d_k}$ since the user do not know the permutation function. This implies that the user $k$ can not obtain any information of the permutation functions for each packet of $W_{d_k}$, i.e.,
\begin{eqnarray*}
I (\{\phi_n\ |\ n\in[0,N)\};\mathcal{Z}_{k},\mathcal{T}_{k},\{X_{\mathcal{U}}|\mathcal{U}\in\mathfrak{U}_k\})=0, \ \ \ \ \ j\in [0,F).
\end{eqnarray*}
Then user $k$ can not obtain any information of any index $j$ for each packet of $W_{d_k}$.
\begin{eqnarray*}
I (\{j \};\mathcal{Z}_{k},\mathcal{T}_{k},\{X_{\mathcal{U}}|\mathcal{U}\in\mathfrak{U}_k\})=0, \ \ \ \ \ j\in [0,F).
\end{eqnarray*}
So the server should send a position message to user $k$ through the secret channel, which only contains the information about the permutation function $\phi_{d_{k}}$. Then the user $k$ can reconstruct the complete file successfully by $\phi_{d_{k}}$. So the following information entropy equation holds.
\begin{eqnarray}
H(W_{d_{k}}|(\mathcal{Z}_{k}, \mathcal{T}_k, \{ T_{\mathcal{U}} \oplus X_{\mathcal{U}} | k \in \mathcal U, \mathcal{U}\in \mathfrak{U}_k \} ,\phi_{d_{k}}))=0.
\end{eqnarray}The entropy information shows that the user $k$ can obtain its requested file once the permutation function $\phi_{d_{k}}$ is received by the user $k$.
\end{itemize}

We impose a privacy constraints on system, i.e., each user can decoded its requested file, but not get any information
about the demands of other $K-1$ users.
 The following side information must hold
\begin{eqnarray}
\label{eq-privacy}
I ({\bf d}_{-k};\{X_{\mathcal{U}}:k \in \mathcal{U}, \mathcal{U}\in \mathfrak{U}_k\}| (\phi_{d_k},\mathcal{Z}_{k},d_k,\mathfrak{U}_k))=0, \ \ \ \ \ k \in [0,K)
\end{eqnarray}
where ${\bf d}_{-k}=(d_{0},d_{1},\ldots,d_{k-1},d_{k+1},\ldots,d_{K-1})$, i,e., the vector consists of all the demands except the demand of the user $k$.
 Since the total size of the position messages is far less than the total size of the encrypted message signals sent by the server, the position messages are negligible in the delivery phase. Hence the total transmission rate $R$ is shorted by
\begin{eqnarray*}
R:=\frac{1}{B}\sum\limits_{\mathcal{U}\in \mathfrak{U}}X_{\mathcal{U}}=\sum\limits_{\mathcal{U}\in \mathfrak{U}}\frac{1}{F}>0.
\end{eqnarray*}
Our objective is minimizing the total transmission rate $R$ by designing the reasonable $F$-division caching scheme while satisfying the privacy constraint of the system, i.e., the equation \eqref{eq-privacy}.

\section{Coded caching schemes with private demands based on secure placement delivery arrays}
\label{Sec-Coded caching scheme-SPDA}
In this section, we will propose a concept of secure placement delivery array (SPDA in short). Using a SPDA, we will  construct an efficient coded caching scheme with the private property when all the requested files are different from each other. For convenience, in the delivery phase we omit the header of each encrypted coded signals and the permutation for each required file sent by the server in the following discussion. As Sengupta et al. in \cite{STC} pointed out that if all the secret keys are chosen mutually independent and randomly from $[0,2)^{B/F}$, the wiretappers reveal no information of the coded signals transmitted in delivery phase. So we shall
 not  consider  the security problems related to the wiretappers in the following.

\subsection{Secure placement delivery array}
\begin{definition}
\label{def-SPDA}
 For positive integers $K$ and $F$, an $F\times K$ array  $\mathbf{P}=(p_{j,k})$, $j \in [0,F)$, $k \in [0,K)$, whose entries are in the set $\{\ast\} \bigcup [0,S)$, is called a $(K,F,Z,S)$ secure placement delivery array (SPDA) if it satisfies the following conditions:
\begin{itemize}
\item[C1.] The symbol "$\ast$" appears $Z$ times in each column;
\item[C2.] Each $s \in [0,S)$ appears at least four times;
\item[C3.] For each integer $s\in [0,S)$, all the rows and the columns having $s$ form the following subarray
\begin{eqnarray}\label{Eqn_Matrix_1}
\mathbf{P}^{(s)}=\left(\begin{array}{cccc}
      s & s & \cdots  & s\\
      s & * & \cdots & *\\
      * & s & \cdots & *\\
      \vdots& \vdots &\ddots & \vdots\\
      * & * & \cdots & s
    \end{array}\right)
\end{eqnarray} with respect to row/column permutations.
\end{itemize}
\end{definition}


\begin{example}
\label{exam-1}
Now let us consider the following array.
\begin{eqnarray}
\label{eq-3-5-SPDA}
    \left(
      \begin{array}{ccc}
        0 & 0 & \ast \\
        1 & \ast & 1 \\
        \ast & 2 & 2 \\
        \ast & 0 & 1 \\
        0 & \ast & 2 \\
        1 & 2 & \ast \\
      \end{array}
    \right)
    \end{eqnarray}
When $s=0$, $1$ and $2$, we have
\begin{eqnarray*}
    \left(
      \begin{array}{cc}
        0 & 0\\
        \ast & 0  \\
        0 & \ast
      \end{array}
    \right)\ \
    \left(
      \begin{array}{cc}
        1 &  1 \\
        \ast  & 1 \\
        1 &  \ast
      \end{array}
    \right)\ \
    \left(
      \begin{array}{cc}
        2 & 2\\
        \ast & 2 \\
        2 & \ast
      \end{array}
    \right).
\end{eqnarray*} So from Definition \ref{def-SPDA} the above array in \eqref{eq-3-5-SPDA} is a $(3,6,2,3)$ SPDA.
\end{example}

\subsection{Private coded caching scheme realized by SPDA}
\label{sub-SPDA-CCS}
Based on a $(K,F,Z,S)$ SPDA $\mathbf{P}=(p_{j,k})$ with $j\in [0,F)$ and $k\in[0,K)$, an $F$-division caching scheme for a $(K,M,N)$ caching system with $M/N=\frac{F+Z}{2F}$, transmission rate $R=S/F$ can be  obtained as follows:
\begin{itemize}
\item [1.] \textbf{Placement Phase:} For each integer $s\in [0,S)$, let $\mathcal{U}_s$ denote the set of columns containing $s$.
User $k$ belongs to $\mathcal{U}_s$ if and only if the $k$th column contains $s$.
Each user $k$ is cached the following secret keys.
  \begin{align}
  \label{eq-key}
  \mathcal{T}_k=\{T_{\mathcal{U}_s}\ |\ k\in \mathcal{U}_{s},s \in[0,S)\}
  \end{align}

  The server randomly chooses $N$ permutation functions $\Omega=\{\phi_{n}\ |\  n\in[0,N)\}$ on $[0,F)$, and then split each file $W_n$ into $F$ packets with equal size, $i.e.$ $W_n=\{W_{n,\phi_{n}(j)}\ |\ j\in[0,F)\}$, $n\in[0,N)$. Let $j^{(s)}$ denote the row containing $s$ at least twice (for each $s$, there is exactly one such a row).
   The server places the following packets into the user $k$'s caches by means of the permutation function $\phi_{n}$,
\begin{eqnarray}
 \label{eq_p1}
  \mathcal{Z}_k=\{W_{n,\phi_{n}(j)}\ |\ p_{j,k}=*, j\in [0,F),  n\in[0,N)\} \bigcup \{W_{n,\phi_{n}(j^{(s)})}|p_{j^{(s)},k}=s,\ s\in [0,S)\}
\end{eqnarray}
  Each user stores $N\cdot (Z+\frac{F-Z}{2})$ packets. Since each packet has size $B/F$, user $k$ caches $N\cdot (Z+\frac{F-Z}{2})\cdot \frac{B}{F}=N\cdot \frac{F+Z}{2}\cdot \frac{B}{F}=MB$ bits.
\item [2.] \textbf{Delivery Phase:} Once the server receives the request vector $\mathbf{d}=(d_0,d_1,\cdots,d_{K-1})$, at the  time slot $s$, $0\le s<S$, it sends $T_{\mathcal{U}_{s}}\bigoplus X_{\mathcal{U}_s}$ to $\mathcal{U}_s$, where
\begin{eqnarray}\label{eq_p2}
      X_{\mathcal{U}_s}=\left(\bigoplus_{p_{j,k}=s,j\in [0,F)\setminus\{j^{(s)},k\in \mathcal{U}_{s}}W_{d_k,\phi_{d_k}(j)}\right)\bigoplus\left(\bigoplus_{n\in [0,N)\setminus \{d_{k},k\in \mathcal{U}_{s}\}}W_{n,\phi_{n}(j^{(s)})}\right).
\end{eqnarray}
\end{itemize}

Assume that in SPDA $\mathbf{P}$ there are $2g$ entries $p_{j^{(s)},k_0}=p_{j^{(s)},k_1}=\cdots=p_{j^{(s)},k_{g-1}}=s$ and $p_{j_0,k_0}=\cdots=p_{j_{g-1},k_{g-1}}=s$ for some positive integer $g>1$. From C3 of Definition \ref{def-SPDA} we have $0\leq j_0,\cdots,j_{g-1}, j^{(s)}< F$, $0\leq k_0,\cdots,k_{g-1}<K$ and the following sub-array formed by rows $j_0$, $\cdots$, $j_{g-1}$, $j^{(s)}$ and columns $k_0$, $\cdots$, $k_{g-1}$.
\begin{eqnarray}\label{Eqn_Matrix_2}
\begin{array}{c|cccc}
       &k_0   & k_1  &\cdots &k_{g-1}\\   \hline
j^{(s)}& s    & s    &\cdots & s\\
j_0    & s    & *    &\cdots & *\\
j_1    &  *   & s    &\cdots & *\\
\vdots &\vdots&\vdots&\ddots & \vdots\\
j_{g-1}& *    & *    &\cdots & s
    \end{array}
\end{eqnarray} According to \eqref{eq_p2}, at each time slot $s\in[0,S)$, the sever sends
\begin{eqnarray*}
      T_{\mathcal{U}_{s}}\bigoplus
      \left(\left( \bigoplus_{0\leq l< g}W_{d_{k_l},\phi_{d_{k_l}}(j_{l})} \right) \bigoplus \left(\bigoplus_{n\in [0,N)\setminus \{d_{k},k\in \mathcal{U}_{s}\}}W_{n,\phi_{n}(j^{(s)})}\right)\right).
\end{eqnarray*}
In \eqref{Eqn_Matrix_2}, for each column $v$, $0\leq v <g$, all the entries are stars except for the $0$-th one and the $(v+1)$-th one.
By \eqref{eq_p1} this implies that user $k_v$ has already cached all the other packets $W_{d_{k_l},\phi_{d_k}(j_{l})}$, $1\leq l\neq (v+1)< g$ and all the packets $W_{n,\phi_{n}(j^{(s)})},n \in [N]$ in the placement phase. At the same time, user $k_v$ cached the secret key $T_{\mathcal{U}_s}$, $k_v \in \mathcal{U}_{s}$. Then it can easily decrypt the signal the server transmitted at the time slot $s$ and decode the desired packet $W_{d_{k_v},\phi_{d_k}(j_v)}$. Since the server sends $S$ packets for each possible request $\mathbf{d}$, the rate of the scheme is given by $S/F$.
\begin{example}
Using the $(3,6,4,3)$ SPDA in \eqref{eq-3-5-SPDA}, we can obtain a $6$-division $(3,\frac{8}{3},4)$ coded caching scheme in the following way.
\begin{itemize}
  \item Placement phase. The server randomly chooses $4$ secret permutation function, $\phi_{n}$, $n\in[0,4)$, which is only know for the server, and splits each file into $6$ packets, denoted by $W_{n}=\{W_{n,\phi_{n}(0)},\ldots,W_{n,\phi_{n}(5)}\}$, $n \in [0,4)$.
By \eqref{eq-key}, users cache the following following secret keys.
\begin{align*}
\mathcal{T}_{0}=\{T_{\{0,1\}},T_{\{0,2\}}\}, \ \ \ \ \mathcal{T}_{1}=\{T_{\{0,1\}},T_{\{1,2\}}\}, \ \ \ \ \mathcal{T}_{2}=\{T_{\{0,2\}},T_{\{1,2\}}\}
\end{align*}
Clearly $\mathfrak{U}=\{\{0,1\},\{0,2\},\{1,2\}\}$. By \eqref{eq_p1},  the users cache the following packets.
  \begin{align*}
  \mathcal{Z}_{0}=\{W_{n,\phi_{n}(0)},W_{n,\phi_{n}(1)},W_{n,\phi_{n}(2)},
      W_{n,\phi_{n}(3)}|n \in [0,4)\}\\
  \mathcal{Z}_{1}=\{W_{n,\phi_{n}(0)},W_{n,\phi_{n}(1)},W_{n,\phi_{n}(2)},
      W_{n,\phi_{n}(4)}|n \in [0,4)\}\\
  \mathcal{Z}_{2}=\{W_{n,\phi_{n}(0)},W_{n,\phi_{n}(1)},W_{n,\phi_{n}(2)},
      W_{n,\phi_{n}(5)}|n \in [0,4)\}
  \end{align*}

\item Delivery phase. Assume that the request vector $\mathbf{d}=(0,1,2)$. Then for each subset $\mathcal{U}\in \mathfrak{U}$, the server transmits $T_{\mathcal{U}_s}\bigoplus X_{\mathcal{U}_s}$ as follows,
\begin{align*}
s=0:\ T_{\{0,1\}}\bigoplus W_{0,\phi_{0}(4)}\bigoplus W_{1,\phi_{1}(3)}\bigoplus W_{2,\phi_{2}(0)}\bigoplus W_{3,\phi_{3}(0)}\\
s=1:\ T_{\{0,2\}}\bigoplus W_{0,\phi_{0}(5)}\bigoplus W_{1,\phi_{1}(1)}\bigoplus W_{2,\phi_{2}(3)}\bigoplus W_{3,\phi_{3}(1)}\\
s=2:\ T_{\{1,2\}}\bigoplus W_{0,\phi_{0}(2)}\bigoplus W_{1,\phi_{1}(5)}\bigoplus W_{2,\phi_{2}(4)}\bigoplus W_{3,\phi_{3}(2)}
\end{align*}
\end{itemize}
We can check that each user will be able to decode its required files. Furthermore, it is easy to check that each coded signal has exact one packet of each file. So for each user $k$, it can not get any information of the other users' demands.
\end{example}

From the above discussion , Algorithm \ref{alg:SPDA}, which can be used to realize a coded caching scheme by a SPDA, is obtained. Moreover, we claim that such a scheme satisfies our security setting for all the users when all the requested files are different from each other. 

\par\begin{algorithm}[http!]
\caption{A secure coded caching scheme based on SPDA}\label{alg:SPDA}
\LinesNumbered
\KwIn{Placement$(\mathbf{P}$, $\mathcal{W})$}
Chooses $N$ permutation functions $\Omega=\{\phi_{n}\ |\  n\in[0,N)\}$ on $[0,F)$ randomly\;
Split each file $W_n\in\mathcal{W}$ into $F$ packets, i.e., $W_{n}=\{W_{n,\phi_{n}(j)}\ |\ j\in[0,F)\}$\;
Let $j^{(s)}$ be the row number containing $s$ at least twice for each $s\in[0,S)$\;
\For{$k\in \mathcal{K}$}{
$\mathcal{Z}_k\leftarrow\mathcal{Z}_k=\{W_{n,\phi_{n}(j)}\ |\ p_{j,k}=*, j\in [0,F),  n\in[0,N)\} \bigcup \{W_{n,\phi_{n}(j^{(s)})}|p_{j^{(s)},k}=s,\ s\in [0,S)\}$
}
\KwIn{Delivery$(\mathbf{P}, \mathcal{W},{\bf d})$}
\For{$s\in[0,S)$}{
    Denote column set each of which contains $s$ by $\mathcal{U}_{s}$\;
    Server broadcasts the following encrypted coded signal\;
    $ T_{\mathcal{U}_{s}}\bigoplus X_{\mathcal{U}_s}=T_{\mathcal{U}_{s}}\bigoplus\left(\bigoplus\limits_{p_{j,k}=s,j\in [0,F)\setminus\{j^{(s)}\},k\in \mathcal{U}_{s}}W_{d_k,\phi_{d_k}(j)}\right)\bigoplus\left(\bigoplus\limits_{n\in [0,N)\setminus \{d_{k},k\in \mathcal{U}_{s}\}}W_{n,\phi_{n}(j^{(s)})}\right)$\;
}
\end{algorithm}
\begin{theorem}
\label{th-main}
Given a $(K,F,Z,S)$ SPDA, using Algorithm \ref{alg:SPDA} we can obtain a $(K,M,N)$ coded caching scheme with $M/N=\frac{F+Z}{2F}$, transmission rate $R=\frac{S}{F}$ and the subpacketization $F$ which satisfies private demand for any user, when all the requested files are different from each other.
\end{theorem}
\begin{proof}Assume that $\mathbf{P}$ is a $(K,F,Z,S)$ SPDA. From Algorithm \ref{alg:SPDA}, all the user can get the required file and the transmission rate is $R=S/F$. When all the requested files are random and different from each other, let us consider the privacy constraint of the scheme. As we have seen, the coded caching scheme can be obtained from a $(K,F,Z,S)$ SPDA by Algorithm \ref{alg:SPDA}. So we only need to show the equation \eqref{eq-privacy} is satisfied. For each user $k$, $k\in[0,K)$, we show the privacy holds, i.e.,
\begin{eqnarray}
\label{eq-huxinxi}
&&I ({\bf d}_{-k};\{X_{\mathcal{U}}\ |\ k \in \mathcal{U},\mathcal{U}\in \mathfrak{U}_k\}|(\phi_{d_k},\mathcal{Z}_{k},d_{k},\mathfrak{U}_k)) =0
\end{eqnarray}We now focus on one user $k$, one demand vector $d_k$, one cache realization $\mathcal{Z}_k$ and one permutation function $\phi_{d_k}$. For each $s\in[0,S)$, given $d_k$, $\mathcal{Z}_k$, $\phi_{d_k}$ and $\mathfrak{U}_k$. Clearly \eqref{eq-huxinxi} equals $0$ if for any demand ${\bf d}_{-k}$, the following probability
\begin{eqnarray}
\label{eq-zhiji}
\begin{split}
&Pr\left(\{X_{\mathcal{U}_s}|k \in \mathcal{U}_{s},s\in [0,S),n\in[0,N)\} |({\bf d}_{-k},\phi_{d_k},\mathcal{Z}_{k},d_{k},\mathfrak{U}_k)\right)
\end{split}
\end{eqnarray}does not depend on ${\bf d}_{-k}$.
By \eqref{eq_p2}, each coded signal $X_{\mathcal{U}_s}$, $k \in \mathcal{U}_{s}$, $s\in [0,S)$ contains exactly $N$ packets each of which belongs to a unique file. Denote the packet of the $n$-th file by $X_{\mathcal{U}_s,n}$ from $X_{\mathcal{U}_s}$ and the index of packet $X_{\mathcal{U}_s,n}$ by $j_{\mathcal{U}_s,n}$. Since the contents of each file is independent to the demands from the users and the content of $X_{\mathcal{U}_{s,n}}$ is correspondent to $\phi_n$,
\eqref{eq-zhiji} can be written as follows.
\begin{eqnarray}
\label{eq-inror-1}
\begin{split}
&Pr\left(\{X_{\mathcal{U}_s,n}|k \in \mathcal{U}_{s},s\in [0,S),n\in[0,N)\}|({\bf d}_{-k},\phi_{d_k},\mathcal{Z}_{k},d_{k},\mathfrak{U}_k)\right)\\
=&Pr\left(\{\phi_n(j_{\mathcal{U}_s,n})|k \in \mathcal{U}_{s},s\in [0,S),n\in[0,N)\}|({\bf d}_{-k},\phi_{d_k},\mathcal{Z}_{k},d_{k},\mathfrak{U}_k)\right)\\
=&\prod\limits_{n\in[0,N)}Pr\left(\{\phi_n(j_{\mathcal{U}_s,n})|k \in \mathcal{U}_{s},s\in [0,S)\}|({\bf d}_{-k},\phi_{d_k},\mathcal{Z}_{k},d_{k},\mathfrak{U}_k)\right)
\end{split}
\end{eqnarray}
 The last equation holds since the permutation functions $\phi_0$, $\phi_1$, $\ldots$, $\phi_{N-1}$ in the placement are independent. Now we consider the probabilities of events $$\{\phi_{d_k}(j_{\mathcal{U}_s,d_k})|k \in \mathcal{U}_{s},s\in [0,S)\}$$ and $$\{\phi_n(j_{\mathcal{U}_s,n})|k \in \mathcal{U}_{s},s\in [0,S)\},\ \ \ n\in [0,N)\setminus\{d_k\}$$
 for given ${\bf d}_{-k}$, $\phi_{d_k}$, $\mathcal{Z}_{k}$, $d_{k}$ and $\mathfrak{U}_k$, respectively, according to their probability space.
\begin{itemize}
  \item When $n=d_{k}$, there are $F-(Z+\frac{F-Z}{2})=\frac{F-Z}{2}$ packets of the file $W_{d_k}$ which is not in user $k$'s cache. Without lose of generality we assume that such the indices of such packets are $\Lambda_{d_k}=\{\phi_{d_k}(0)$, $\phi_{d_k}(1)$, $\ldots$, $\phi_{d_k}((F-Z)/2-1)\}$. In addition, each required packet occurs in exact one coded signal of $(X_{\mathcal{U}_s}|k \in \mathcal{U}_{s},s\in [0,S))$. So all the permutation of $\Lambda_{d_k}$ is the probability space of random event $(\phi_{d_k}(j_{\mathcal{U}_s,d_k})|k \in \mathcal{U}_{s},s\in [0,S))$. Denote all the permutation combinatorial functions of $\Lambda_{d_k}$ by $\{\psi_h\ |\  h\in[0,((F-Z)/2)!)\}$. Then for each $h\in [0,((F-Z)/2)!)$, we have
      \begin{eqnarray}
      \label{eq-prb-1}
  Pr\left(\left\{\phi_{d_k}(j_{\mathcal{U}_s,d_k})|k \in \mathcal{U}_{s},s\in [0,S)\right\}=\psi_h\left(\Lambda_{d_k})\right)|({\bf d}_{-k},\phi_{d_k},\mathcal{Z}_{k},d_{k},\mathfrak{U}_k)\right)
  =\frac{1}{\left(\frac{F-Z}{2}\right)!}.
  \end{eqnarray}
  \item When $n\neq d_{k}$, there are $Z+\frac{F-Z}{2}=\frac{F+Z}{2}$ packets of the file $W_{n}$, $n\in [0,N)\setminus\{d_k\}$ cached by user $k$. Similarly denote these indices of such packets by $\Lambda_{n}=\{\phi_{n}(0)$, $\phi_{n}(1)$, $\ldots$, $\phi_{d_k}((F+Z)/2-1)\}$. In addition, each required chosen packet occurs in exact one coded signal of $(\{X_{\mathcal{U}_s}|k \in \mathcal{U}_{s},s\in [0,S)\}$. So all the permutation of  any $(F+Z)/2$ indices of $\Lambda_{n}$ is the probability space of random event $(\phi_{n}(j_{\mathcal{U}_s,d_k})|k \in \mathcal{U}_{s},s\in [0,S))$. Denote all the permutation combinatorial functions of $\Lambda_{n}$ by $\{\mu_h\ |\ h\in [0,{(F+Z)/2\choose (F-Z)/2}((F-Z)/2)!)\}$. Then for each $h\in [0,{(F+Z)/2\choose (F-Z)/2}((F-Z)/2)!)$, we have
  \begin{eqnarray}
  \label{eq-prb-2}
  Pr\left(\left\{\phi_{n}(j_{\mathcal{U}_s,n})|k \in \mathcal{U}_{s},s\in [0,S))=\mu_h(\Lambda_{n})\right\}|({\bf d}_{-k},\phi_{d_k},\mathcal{Z}_{k},d_{k},\mathfrak{U}_k)\right) =\frac{1}{{(F+Z)/2\choose (F-Z)/2}\left(\frac{F-Z}{2}\right)!}
  \end{eqnarray}
\end{itemize}
Submitting \eqref{eq-prb-1} and \eqref{eq-prb-2} into \eqref{eq-zhiji} and by \eqref{eq-inror-1} we can see that
\begin{eqnarray*}
&&Pr\left(\{X_{\mathcal{U}_s}|k \in \mathcal{U}_{s},s\in [0,S),n\in[0,N)\} |({\bf d}_{-k},\phi_{d_k},\mathcal{Z}_{k},d_{k},\mathfrak{U}_k)\right)\\
&=&Pr\left(\left\{\phi_n(j_{\mathcal{U}_s,n})|k \in \mathcal{U}_{s},s\in [0,S),n\in[0,N)\right\}|({\bf d}_{-k},\phi_{d_k},\mathcal{Z}_{k},d_{k},\mathfrak{U}_k)\right)\\
&=&\frac{1}{\left(\frac{F-Z}{2}\right)!}\cdot \left( \frac{1}{{(F+Z)/2\choose (F-Z)/2}\left(\frac{F-Z}{2}\right)!}\right)^{N-1}.\\
\end{eqnarray*}
is an constant number for any ${\bf d}_{-k}$. This implies that for  given $d_k$, $\mathcal{Z}_k$, $\phi_{d_k}$ and $\mathfrak{U}_k$, the event $\{X_{\mathcal{U}_s}|k \in \mathcal{U}_{s},s\in [0,S),n\in[0,N)\}$ and the event ${\bf d}_{-k}$ are independent.

Then the proof is completed.
\end{proof}
\begin{remark}
For a $(K,F,Z,S)$ SPDA $\mathbf{P}$,  the following statements give further explanations about the related  coded caching scheme.
\begin{itemize}
\item If entry $p_{j,k}=*$ or $p_{j^{(s)}}=s$ for each $s\in[0,S)$, then user $k$ has already cached the packets induced by $j$ and $j^{(s)}$  of  each of all the files in server. The server sums of all the packets indexed by $j^{(s)}$ of all the unrequired files, i.e., $X'_{\mathcal{U}_s}=\bigoplus_{n\in [0,N)\setminus \{d_{k},k\in \mathcal{U}_{s}\}}W_{n,\phi_{n}(j^{(s)})}$.
\item If entry $p_{j,k}=s$  and $s$ occurs exact once in the $j$th row, it means that all the packets induced by $j$ of all the files are not stored by user $k$. Then the server looks at all the such row numbers $j$ of $\mathbf{P}$. The server gets the sum of all the the required packets , i.e., $X''_{\mathcal{U}_s}=\bigoplus_{p_{j,k}=s,j\in [0,F)\setminus\{j^{(s)}\},k\in \mathcal{U}_{s}}W_{d_k,\phi_{d_k}(j)}$, and gets the coded signal $X_{\mathcal{U}_s}=X''_{\mathcal{U}_s}\bigoplus X'_{\mathcal{U}_s}$. The third property of the SPDA guarantees the user can get exactly one requested packet from encrypted coded signal.
\item In this paper, we only consider the coded signals containing multicast messages. So each integer must occur at least four times.
\end{itemize}
\end{remark}
\begin{remark}
From Theorem \ref{th-main}, we can see that each coded signal contains exact one packet of each file, and each packet occurs at most once in all the coded signals. Then a user can decode a required packet from a coded signal only if it has already cached the other packets from this coded signal. So the number of packet for each file cached by the user $k$ must be larger than or equal to the number of packet for each file uncached by user $k$. This implies that the memory ratio must be larger than or equal to $1/2$. If the memory ratio is less than $1/2$, we can divide each file into two parts. Assume that the first part consists of the previous $B_1$ bits of each file. When the ratio $MB/(NB_1)>1/2$, we can apply SPDA for the first part. For the second part, the server just broadcasts each part of the required files with unique secret key for each user.
\end{remark}

From Theorem \ref{th-main},  an $F$-division $(K,M,N)$ secure coded caching scheme can be obtained by constructing an appropriate SPDA when all the requested files are distinct. So we only focus on constructing SPDAs in the following.
\section{The minimum transmission rate}
\label{Lower bounds}
From Theorem \ref{th-main}, given the subpacketization $F$, user number $K$ and the memory ratio $M/N$, the transmission rate $R$ of the secure coded caching scheme realized by a $(K,F,Z,S)$ SPDA is determined by the value of $S$. So it is meaningful to consider the lower bound of the the parameter $S$ for a $(K,F,Z,S)$ PDA. In this section, the lower bound of $S$ in a SPDA is derived by extremal combinatorial method.
\begin{lemma}
\label{le-lower bound}
If there exists a $(K,F,Z,S)$ SPDA $\mathbf{P}$ with integer set $[0,S)$, then
\begin{eqnarray}
\label{eq-Lowebound-K<F-1}
\frac{K(F-Z)}{2+2K-\frac{2 K(F-Z)}{F+Z}}\leq S.
\end{eqnarray}%
\end{lemma}
\begin{proof}
Assume that an integer $s\in [0,S)$ occurs exactly $r_s$ times in $\mathbf{P}$. Then the subarray $\mathbf{P}^{(s)}$ in \eqref{Eqn_Matrix_1} has $\frac{r_s}{2}+1$ rows, $\frac{r_s}{2}$ columns and exactly $\frac{r_s}{2}(\frac{r_s}{2}-1)$ stars. Then the total number of stars in $\mathbf{P}^{(s)}$, $s=0$, $1$, $\ldots$, $S-1$, is
$$W=\sum\limits_{s=0}^{S-1}\frac{r_s}{2}(\frac{r_s}{2}-1)$$

On the other hand, we can estimate the value of $W$ in a different way. For each integer $i\in[0,F)$, we assume that row $i$ has exactly $x_i+y_i$ integer entries where there are $x_i$ integers entries of which occurs at least twice, say $s_{i,1}$, $\ldots$, $s_{i,x_i}$ and $y_i$ integer entries of which occurs exact once, say $s_{i,x_i+1}$, $\ldots$, $s_{i,x_i+y_i}$. By \eqref{Eqn_Matrix_1} we have that each star of the $i$-th row occurs in at most $y_i$ times in $\mathbf{P}^{(s_{i,x_i+1})}$, $\ldots$, $\mathbf{P}^{(s_{i,x_i+y_i})}$ since it occurs in $\mathbf{P}^{(s_{i,h})}$ at most once for any $h\in (x_i,x_i+y_i]$.
Then the total number of occurrences of all the stars of $\mathbf{P}$ in $\mathbf{P}^{(s)}$, $s=0$, $1$, $\ldots$, $S-1$, is at most $$W'=\sum\limits_{i=0}^{F-1}y_i(K-x_i-y_i).$$
Clearly $W\leq W'$, i.e., $\sum_{s=0}^{S-1}\frac{r_s}{2}(\frac{r_s}{2}-1)\leq\sum_{i=0}^{F-1}y_i(K-x_i-y_i)$.
Then we have
\begin{eqnarray}
\label{eq-inquality}
\sum\limits_{s=0}^{S-1} r^2_s+4\sum\limits_{i=0}^{F-1}y^2_i \leq 2\sum\limits_{s=0}^{S-1}r_s+4K\sum\limits_{i=0}^{F-1}y_i-4\sum\limits_{i=0}^{F-1}x_i y_i.
\end{eqnarray} In fact we can get the number of integer entries by counting the occurrence number of each integer $s$ and the integer entries in each row respectively, i.e.,
$$K(F-Z)=\sum_{s=0}^{S-1}r_s= \sum_{i=0}^{F-1}(x_i+y_i).$$ Then \eqref{eq-inquality} can be written as
\begin{eqnarray}
\label{eq-IN 1}
\sum\limits_{s=0}^{S-1}r^2_s+4\sum\limits_{i=0}^{F-1}{y}^2_i\leq 4K\sum\limits_{i=0}^{F-1}y_i-4\sum\limits_{i=0}^{F-1}x_iy_i+2K(F-Z).
\end{eqnarray}
Moreover, by the convexity
\begin{eqnarray*}
\label{eq-IN 2}
\begin{split}
&\sum\limits_{s=0}^{S-1}r^2_s\geq\frac{1}{S}\left(\sum\limits_{s=0}^{S-1}r_s\right)^2=\frac{K^2(F-Z)^2}{S},
\end{split}\end{eqnarray*}
\eqref{eq-IN 1} can be written as
\begin{eqnarray}
\label{eq-IN 3}
\frac{K^2(F-Z)^2}{S}+4\sum\limits_{i=0}^{F-1}{y}^2_i\leq 2K(F-Z)+4K\sum\limits_{i=0}^{F-1}y_i-4\sum\limits_{i=0}^{F-1}x_iy_i.
\end{eqnarray}
From the property C2 of SPDA, $\sum\limits_{i=0}^{F-1}y_i>0$ always holds. So $S$ can get the minimize value if and only if
$$\sum\limits_{i=0}^{F-1}x_iy_i=0.$$
This implies that for each row $i$, $x_i=0$ or $y_i=0$ always holds, i.e., in each row all the integers occuring at least twice or occuring exactly once. Without loss of generality, denote the rows where each integer entry occurs exactly once by $ [0,F')$. We have
\begin{eqnarray}
\label{eq-IN 2-1}
K(F-Z)=2K\sum\limits_{i=0}^{F'-1}y_i.
\end{eqnarray}Submitting the following convexity
\begin{eqnarray*}
\label{eq-IN 2}
\begin{split}
&\sum\limits_{i=0}^{F'-1}{y}^2_i\geq\frac{1}{F'}\left(\sum\limits_{i=0}^{F'-1}y_i\right)^2=\frac{K^2(F-Z)^2}{4F'}
\end{split}\end{eqnarray*} and \eqref{eq-IN 2-1} into \eqref{eq-IN 3}, we have
\begin{eqnarray}
\label{eq-IN 4}
\frac{K(F-Z)}{S}+\frac{K(F-Z)}{F'}\leq 2+2K.
\end{eqnarray}
where the equalities hold if and only if $r_0=r_1=\ldots=r_{S-1}$ and $y_0=y_1=\ldots=y_{F'-1}$ respectively. Clearly the number of the rows where all the integer entries occur at least twice is at least $\frac{K(F-Z)}{2K}=\frac{F-Z}{2}$. In this case $F'=F-\frac{F-Z}{2}$. Then \eqref{eq-IN 4} can be written as
\begin{eqnarray*}
\frac{K(F-Z)}{S}\leq 2+2K-\frac{2K(F-Z)}{F+Z}.
\end{eqnarray*}
Then the inequality \eqref{eq-Lowebound-K<F-1} can be obtained. And the equality in \eqref{eq-Lowebound-K<F-1} holds if and only if
\begin{itemize}
\item $r_0=r_1=\ldots=r_{S-1}=\frac{K(F-Z)}{S}$,
\item $x_0$ $=x_1$ $=\ldots$ $=x_{\frac{F+Z}{2}-1}=0$, $y_0$ $=y_1$ $=\ldots$ $=y_{\frac{F+Z}{2}-1}=\frac{F-Z}{2}$, and  are positive integers.
\item $x_{\frac{F+Z}{2}}$ $=x_{\frac{F+Z}{2}+1}$ $=\ldots$ $=x_{F-1}=K$, $y_{\frac{F+Z}{2}}$ $=y_{\frac{F+Z}{2}+1}$ $=\ldots$ $=y_{F-1}=0$.
\end{itemize}%
The proof is completed.
\end{proof}
From Lemma \ref{le-lower bound}, the following result can be obtained.
\begin{theorem}
\label{th-lower-bound}
If a $(K,M,N)$ secure coded caching scheme can be realized by a $(K,F,Z,S)$ SPDA with $M/N=\frac{F+Z}{2F}$, then its minimum transmission rate
\begin{eqnarray}
\label{eq-Lowebound-R}
R=\frac{S}{F}\geq \frac{1}{2}\cdot\frac{K}{1+\frac{Z}{F}+2K\frac{Z}{F}}\left(1-\frac{Z}{F}\right).
\end{eqnarray}%
\end{theorem}

For the fixed parameters $K$, $F$ and $Z$, a $(K,F,Z,S)$-SPDA with the minimum $S$ is called optimal. In the following we will propose three class of constructions of SPDAs, where one of them is optimal.
\section{Constructions of SPDAs}
\label{sec-construciton}
In this section, three classes of SPDAs are constructed directly where the first class is optimal. First the following construction is necessary for our constructions.
\begin{construction}
\label{const-PDA}
For any positive integers $K$ and $t$ with $t<K$, let $\mathcal{F}={[0,K]\choose t}$, $\mathcal{K}=[0,K)$. We can define a ${K\choose t}\times K$ array $\mathbf{Q}=(q_{A,k})_{A\in \mathcal{F},k\in \mathcal{K}}$
by
\begin{eqnarray}\label{Eqn_Def_AN}
q_{A,k}=\left\{\begin{array}{cc}
A\cup\{k\}, & \mbox{if}~k\notin A\\
*, & \mbox{otherwise}
\end{array}
\right.
\end{eqnarray}It is easy to check that there are exact ${K-1\choose t-1}$ stars in each column of $\mathbf{Q}$.
\end{construction}
\begin{example}\label{ex-(4,4,1,6) PDA}
When $K=4$ and $t=1$, by \eqref{Eqn_Def_AN} the following array, denoted by $\mathbf{Q}_{4\times 4}$, can be obtained.
\begin{eqnarray*}
\begin{array}{c|cccc}
     &0      &1      &2      &3\\ \hline
\{0\}&*      &\{0,1\}&\{0,2\}&\{0,3\}\\
\{1\}&\{0,1\}&*      &\{1,2\}&\{1,3\}\\
\{2\}&\{0,2\}&\{1,2\}&*      &\{2,3\}\\
\{3\}&\{0,3\}&\{1,3\}&\{2,3\}&*
\end{array}
\end{eqnarray*}
\end{example}%
%
%
%
%
%
%

\subsection{The case $(t+1)|K$}
\label{sub-first-case}
It is not difficult to see that when $(t+1)|K$, all elements in ${[0,K)\choose t+1}$ can be partitioned into ${K-1\choose t}$ parallel classes, say $\mathcal{D}_0,\mathcal{D}_1,\ldots,\mathcal{D}_{{K-1\choose t}-1}$, such that for all $j \in [0,{K-1\choose t})$, i) if $B,B' \in \mathcal{B}_j$, then $B\bigcap B'=\emptyset$, and ii) $\bigcup_{B\in \mathcal{B}_j}B=[0,K)$.
\begin{construction}
\label{const-1}
For any positive integers $K$ and $t$ with $(t+1)|K$, let $\mathcal{F}'=\{\mathcal{D}_0,\mathcal{D}_1,\ldots,\mathcal{D}_{{K-1\choose t}-1}\}$ and $\mathcal{K}=[0,K)$. Define a ${K-1\choose t}\times K$ array $\mathbf{Q}'=(q'_{\mathcal{D},k})_{\mathcal{D}\in\mathcal{F}',k\in \mathcal{K}}$
where $q'_{\mathcal{D},k}=D$ if $k\in D$ and $D\in \mathcal{D}$.
\end{construction}
\begin{example}\label{ex-(4,7,1,6) SPDA}
We consider the parameters in Example \ref{ex-(4,4,1,6) PDA} again, i.e., $K=4$ and $t=1$. Let $\mathcal{D}_0=\{\{0,1\},\{2,3\}\}$, $\mathcal{D}_1=\{\{0,2\},\{1,3\}\}$, $\mathcal{D}_2=\{\{0,3\},\{1,2\}\}$. From Construction \ref{const-1} the following array, denoted by $\mathbf{Q}'_{3\times 4}$, can be obtained.
\begin{eqnarray*}
\begin{array}{c|cccc}
             &0      &1      &2      &3\\ \hline
\mathcal{D}_0&\{0,1\}&\{0,1\}&\{2,3\}&\{2,3\}\\
\mathcal{D}_1&\{0,2\}&\{1,3\}&\{0,2\}&\{1,3\}\\
\mathcal{D}_2&\{0,3\}&\{1,2\}&\{1,2\}&\{0,3\}
\end{array}
\end{eqnarray*}
Replace the subsets $\{0,1\},\{0,2\}$, $\{0,3\}$, $\{1, 2\}$, $\{1, 3\}$ and $\{2, 3\}$ in the above array by $0$, $1$, $2$, $3$, $4$ and $5$ respectively. We have
\begin{eqnarray*}
\mathbf{Q}'_{3\times 4}=\left(\begin{array}{cccc}
0&0&5&5\\
1&4&1&4\\
2&3&3&2
\end{array}\right).
\end{eqnarray*}
It is easy to check that
\begin{eqnarray}
\label{eq-optimal}
\mathbf{P}=\left(\begin{array}{c}
\mathbf{Q}'_{3\times 4}\\
\mathbf{Q}_{4\times 4}
\end{array}\right)=\left(\begin{array}{cccc}
0&0&5&5\\
1&4&1&4\\
2&3&3&2\\
*&0&1&2\\
0&*&3&4\\
1&3&*&5\\
2&4&5&*
\end{array}\right).
\end{eqnarray}
is a $(4,7,1,6)$ SPDA. Let $F=7$ and $Z=1$. By \eqref{eq-Lowebound-K<F-1} in Lemma \ref{le-lower bound}, we have
$$
6=S\geq \frac{K(F-Z)}{2+2K-\frac{2 K(F-Z)}{F+Z}}=\frac{4\times(7-1)}{2+8-\frac{2\times 4(7-1)}{7+1}}=6.
$$
So the SPDA in \eqref{eq-optimal} is optimal.
\end{example}
From Example \ref{ex-(4,7,1,6) SPDA}, the following result can be obtained.
\begin{theorem}
\label{th-SPDA-optimal}
For any positive integers $K$ and $t$ with $(t+1)|K$, there always exists an optimal $(K,\frac{2K-t}{K-t}{K-1\choose t}, {K-1\choose t-1},{K\choose t+1})$ SPDA which gives a $(K,M,N)$ secure coded caching scheme with $M/N=\frac{K}{2K-t}$, transmission rate $R=\frac{K-t}{t+1}\cdot\frac{K}{2K-t}$ and the subpacketization $\frac{2K-t}{K-t}{K-1\choose t}$.
\end{theorem}
\begin{proof}
Let us consider the array $\mathbf{P}={ \mathbf{Q}'\choose\mathbf{Q}}$ where $\mathbf{Q}$ is generated by \eqref{Eqn_Def_AN} and $\mathbf{Q}'$ is generated by Construction \ref{const-1}. The number of stars in each column is ${K-1\choose t-1}$ since each column of $\mathbf{Q}$ has ${K-1\choose t-1}$ stars and each column of $\mathbf{Q}'$ does not have star. So $Z={K-1\choose t-1}$. Furthermore
$S={K\choose t+1}$ since each $(t+1)$ subset occurs in $\mathbf{P}$.

Now let us consider the Condition C3 of Definition \ref{def-SPDA}. Without loss of generality let us consider $\mathcal{S}=\{0,1,\ldots, t\}$. The rows of $\mathbf{P}$, which consist of the rows of $\mathbf{Q}$ indexed by the $t$-subsets of $\mathcal{S}$ and a row of $\mathbf{Q}'$ indexed by a partition, say $\mathcal{D}\in \mathcal{F}'$, satisfying $\mathcal{S}\in \mathcal{D}$, and the columns indexed by the elements in $\mathcal{S}$, i.e., $0$, $1$, $\ldots$, $t$, form the following $(t+1)\times t$ subarray.
\begin{eqnarray*}
\begin{array}{c|cccc}
 & 0& 1&\cdots&t\\ \hline
\mathcal{D}& \mathcal{S} & \mathcal{S} & \cdots  & \mathcal{S}\\
\{0,1,\ldots,t-1\} &     \mathcal{S} & * & \cdots & *\\
 \{1,\ldots,t-1,t\}  &   * & \mathcal{S} & \cdots & *\\
 \vdots&     \vdots& \vdots &\ddots & \vdots\\
\{0,\ldots,t-2,t\}&      * & * & \cdots &\mathcal{S}
    \end{array}
\end{eqnarray*}
Clearly the above subarray satisfies C3 of Definition \ref{def-SPDA}. Then $\mathbf{P}$ is a $(K,\frac{2K-t}{K-t}{K-1\choose t}, {K-1\choose t-1},{K\choose t+1})$ SPDA. Submitting the values of $K$, $Z$, $F$ into \eqref{eq-Lowebound-K<F-1} in Lemma \ref{le-lower bound}, we have
\begin{eqnarray*}
{K\choose t+1}=S\geq \frac{K(\frac{2K-t}{K-t}{K-1\choose t}-{K-1\choose t-1})}{2+2K-\frac{2 K(\frac{2K-t}{K-t}{K-1\choose t}-{K-1\choose t-1})}{\frac{2K-t}{K-t}{K-1\choose t}+{K-1\choose t-1}}}=\frac{K{K-1\choose t}}{1+K-(K-t)}={K\choose t+1}.
\end{eqnarray*}
So $\mathbf{P}$ is optimal. Finally from Theorem \ref{th-main}, the desired coded caching scheme can be obtained. Our proof is completed.

\end{proof}
In the following we will propose the other two classes of SPDAs by modifying Construction \ref{const-1}.
\subsection{The case $(t+1)\nmid K$ and $(t+1)\leq\lfloor K/2\rfloor$}

\begin{construction}
\label{const-2}
For any positive integers $K$ and $t$ with $(t+1)\nmid K$ and $(t+1)\leq\lfloor K/2\rfloor$, let $K'$ be the minimum integer satisfying $K'>K$ and $(t+1)|K'$. Let $\mathcal{F}'=\{\mathcal{D}_0,\mathcal{D}_1,\ldots,\mathcal{D}_{{K'-1\choose t}-1}\}$, i.e., the parallel classes, and $\mathcal{K}=[0,K)$. Define a ${K'-1\choose t}\times K$ array $\mathbf{Q}'=(q'_{\mathcal{D},k})_{\mathcal{D}\in \mathcal{F}',k\in \mathcal{K}}$ where
\begin{eqnarray}\label{Eqn_Def-3}
q'_{\mathcal{D},k}=\left\{\begin{array}{cc}
D, & \mbox{if}~k\in D, D\cap [K,K')=\emptyset, D\in \mathcal{D}, \\
*, & \mbox{otherwise}
\end{array}
\right.
\end{eqnarray}
\end{construction}
\begin{example}\label{ex-(4,101,6) SPDA}
Assume that $K=5$ and $t=1$. Then $K'=6$. Let the parallel classes
\begin{eqnarray*}
\begin{array}{ccc}
\mathcal{D}_0=\{\{0,1\},&\{2,3\},&\{4,5\}\}\\
\mathcal{D}_1=\{\{0,2\},&\{3,5\},&\{1,4\}\}\\
\mathcal{D}_2=\{\{0,3\},&\{2,4\},&\{1,5\}\}\\
\mathcal{D}_3=\{\{0,4\},&\{2,5\},&\{1,3\}\}\\
\mathcal{D}_4=\{\{0,5\},&\{3,4\},&\{1,2\}\}
\end{array}
\end{eqnarray*} From Construction \ref{const-2} the following array, denoted by $\mathbf{Q}'_{5\times 5}$, can be obtained.
\begin{eqnarray*}
\begin{array}{c|ccccc}
             &0      &1      &2      &3      &4\\ \hline
\mathcal{D}_0&\{0,1\}&\{0,1\}&\{2,3\}&\{2,3\}&*\\
\mathcal{D}_1&\{0,2\}&\{1,4\}&\{0,2\}&*      &\{1,4\}\\
\mathcal{D}_2&\{0,3\}&*      &\{2,4\}&\{0,3\}&\{2,4\}\\
\mathcal{D}_3&\{0,4\}&\{1,3\}&\{1,3\}&*      &\{0,4\}   \\
\mathcal{D}_4&*      &\{1,2\}&\{1,2\}&\{3,4\}&\{3,4\}
\end{array}
\end{eqnarray*}
\end{example}

We can see that each column of $\mathbf{Q}'$ generated by Construction \ref{const-2} has ${K'-1\choose t}-{K-1\choose t}$ stars. Similar to the proof of Theorem \ref{th-SPDA-optimal}, we claim that the array $\mathbf{P}={\mathbf{Q}'\choose\mathbf{Q}}$ is also a $(K,{K\choose t}+{K'-1\choose t}, {K-1\choose t-1}+{K'-1\choose t}-{K-1\choose t},{K\choose t+1})$ SPDA, where $\mathbf{Q}$ is generated by \eqref{Eqn_Def_AN} and $\mathbf{Q}'$ is generated by Construction \ref{const-2}. Then from Theorem \ref{th-main} the following result can be obtained.
\begin{theorem}
\label{th-SPDA-2}
For any positive integers $K$ and $t$ with $(t+1)\nmid K$ and $\lfloor K/2\rfloor>(t+1)$, let $K'$ is the minimum integer satisfying $K'>K$ and $(t+1)|K'$. There always exists a $(K,{K\choose t}+{K'-1\choose t}, {K-1\choose t-1}+{K'-1\choose t}-{K-1\choose t},{K\choose t+1})$ SPDA which gives a $(K,M,N)$ secure coded caching scheme with $M/N=\frac{{K-1\choose t-1}+{K'-1\choose t}}{{K\choose t}+{K'-1\choose t}}$,
transmission rate $R=\frac{{K\choose t+1}}{{K\choose t}+{K'-1\choose t}}$ and the subpacketization ${K\choose t}+{K'-1\choose t}$.
\end{theorem}
\subsection{The case $(t+1)\nmid K$ and $\lfloor K/2\rfloor<(t+1)$}
\begin{construction}
\label{const-3}
For any positive integers $K$ and $t$ with $(t+1)\nmid K$ and $K/2<(t+1)$, let $\mathcal{F}'={[0,K)\choose t+1}$, $\mathcal{K}=[0,K)$, define a ${K\choose t+1}\times K$ array $\mathbf{Q}'=(q'_{D,k})_{D\in \mathcal{F}',k\in \mathcal{K}}$ where
\begin{eqnarray}\label{Eqn_Def-3}
q'_{D,k}=\left\{\begin{array}{cc}
\mathcal{S}, & \mbox{if}~k\in D\\
*, & \mbox{otherwise}
\end{array}
\right.
\end{eqnarray}
\end{construction}
We can see that each column of $\mathbf{Q}'$ generated by Construction \ref{const-3} has ${K-1\choose t+1}$ stars. Similar to the proof of Theorem \ref{th-SPDA-optimal}, we claim that the array $\mathbf{P}={\mathbf{Q}'\choose\mathbf{Q}}$ is also a $(K,{K\choose t}+{K\choose t+1}, {K-1\choose t-1}+{K-1\choose t+1},{K\choose t+1})$ SPDA, where $\mathbf{Q}$ is generated by \eqref{Eqn_Def_AN} and $\mathbf{Q}'$ is generated by Construction \ref{const-1}. Then the following result can be obtained.
\begin{theorem}
\label{th-SPDA-3}
For any positive integers $K$ and $t$ with $(t+1)\nmid K$ and $\lfloor K/2\rfloor<(t+1)$, there always exists a $(K,{K\choose t}+{K\choose t+1}, {K-1\choose t-1}+{K-1\choose t+1},{K\choose t+1})$ SPDA which gives a $(K,M,N)$ secure coded caching scheme with $M/N=1-\frac{(K-t)(t+1)}{K(K+1)}$, transmission rate $R=\frac{K-t}{K+1}$ and the subpacketization ${K\choose t}+{K\choose t+1}$.
\end{theorem}
\subsection{Performance analyses}
\label{sec-comparison}
As far as we know, the latest scheme with private demands is propose by Wan et al in \cite{WC}, i.e., the scheme listed in Table \ref{tab-known-1}. However from the parameters in Table \ref{tab-known-1} and Table \ref{tab-main}, it is hard to give a comparison between the scheme in \cite{WC} and our scheme in general. So here we take $K=10$ as an example to compare the schemes. In table \ref{tab-main} we have a scheme with private demands for any file number $N$, say Scheme1, where the subpacketization with log function and the transmission rate in red are listed in Figure \ref{Fig-M-F-1} and Figure \ref{Fig-M-R-1} respectively. When $N=50$ and $100$, we also can obtain two schemes, say Scheme2 and Scheme3, from Table \ref{tab-known-1} where the subpacketization with log function and the transmission rate in blue and orange are listed in Figure \ref{Fig-M-F-1} and Figure \ref{Fig-M-R-1} respectively.
\begin{figure}[htbp]
  \centering
  \subfigure[The logarithmic function with subpacketization]{\label{Fig-M-F-1}
  \begin{minipage}{8cm}
  \centering
  \includegraphics[width=1\textwidth]{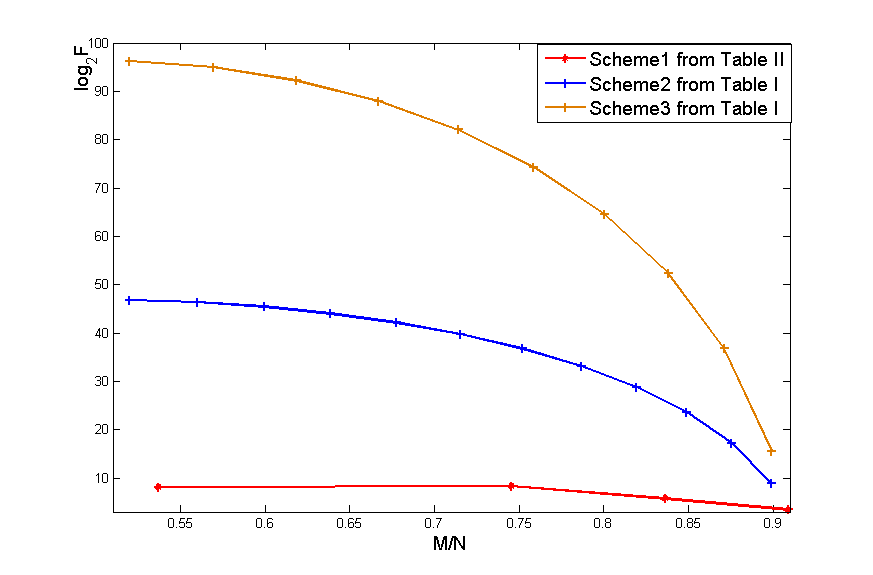}
  \end{minipage}}
  \subfigure[The transmission rate]{\label{Fig-M-R-1}
  \begin{minipage}{8cm}
  \centering
  \includegraphics[width=1\textwidth]{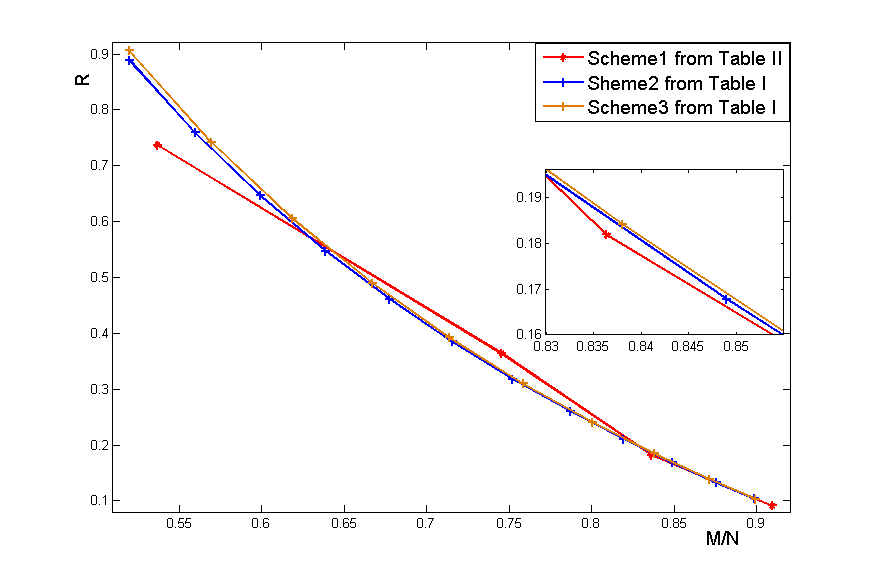}
  \end{minipage}}
  \caption{The performances of Scheme1, Scheme2 and Scheme3 from Tables \ref{tab-main} and \ref{tab-known-1} for $K=10$}
\end{figure}

From Figure \ref{Fig-M-F-1} and Figure \ref{Fig-M-R-1}, we can see that Scheme3 has larger transmission rate and subpacketization than those of Scheme1. So Figure \ref{Fig-M-F-1} and Figure \ref{Fig-M-R-1} intuitively show that when $N$ is larger the performance of the scheme in Table \ref{tab-known-1} is weaker, i.e., $N$ is larger, the subpacketization and the transmission rate of the scheme in Table \ref{tab-known-1} are larger for the fixed user number $K$.

Now let us consider the performance of Scheme1 and Scheme2 listed in Figure \ref{Fig-M-F-1} and Figure \ref{Fig-M-R-1}. By Figure \ref{Fig-M-F-1} the subpacketization of Scheme1 is much smaller than that of Scheme2 while by Figure \ref{Fig-M-R-1} the transmission rates of Scheme1, Scheme2 and Scheme3 are very close. Furthermore when $M_1/N=0.53684$ the transmission rate and the subpacketization of Scheme1 are $R_1=0.70528$ and $F_1=285$. However when $M_2/N=0.53989$ the transmission rate and the subpacketization of Scheme2 are $R_2=0.82162$ and $F_2=108018112524940$. Clearly $M_1<M_2$ but we have $R_1<R_2$ and $F_1\ll F2$. This implies that in this case our scheme has much better performance on both transmission rate and the subpacketization.


\section{Conclusion}
\label{conclusion}
In this paper, we considered the security setting which can prevent the wiretappers from obtaining any information about the files in the library and protect the demands privacy from all the users in delivery phase. We first reformulated this security setting and introduced SPDA to characterize the placement and delivery phases for a coded caching scheme. In order to measure the transmission rate of our secure scheme, the lower bound on the parameter $S$ of a $(K,F,Z,S)$ SPDA was derived. According to our lower bound, we constructed three classes of SPDAs where one of them achieves the lower bound for any positive integer $t$, $K$ with $(t+1)|K$. Consequently three classes of secure schemes with good performance were obtained, where the transmission rate and the subpacketization are independent to the file number $N$. Finally a comparison was proposed to show that our schemes significantly reduce the subpacketizations. It is worth noting that for some parameters, our schemes have smaller transmission rate.

\bibliographystyle{IEEEtran}

\end{document}